\pdfoutput=1
\documentclass[10pt, conference, letterpaper]{IEEEtran}
\usepackage{amsmath,amssymb,amsthm}
\usepackage{graphicx}
\usepackage{tikz}
\usetikzlibrary{arrows.meta,calc}
\usepackage{cite}
\usepackage{textcomp}
\newtheorem{theorem}{Theorem}
\newtheorem{lemma}{Lemma}
\newtheorem{corollary}{Corollary}
\theoremstyle{remark}\newtheorem{remark}{Remark}\theoremstyle{plain}
\begin{document}
\title{S-PASS: Star-Topology Pinching-Antenna Systems with Pipelined Antenna Pre-Positioning}
\author{
\IEEEauthorblockN{Wei Jiang}
\IEEEauthorblockA{Intelligent Networking Research Group\\
German Research Center for Artificial Intelligence (DFKI)\\
Kaiserslautern, 67663 Germany}
\and
\IEEEauthorblockN{Hans D. Schotten}
\IEEEauthorblockA{Department of Electrical and Computer Engineering\\
University of Kaiserslautern (RPTU)\\
Kaiserslautern, 67663 Germany}
}
\maketitle
\raggedbottom
\begin{abstract}
Pinching-antenna systems (PASS) flexibly reconfigure the wireless channel by positioning dielectric antennas along a waveguide near each user. Releasing and re-attaching an antenna, however, is a mechanical operation orders of magnitude slower than electronic switching, wasting much transmission time. Existing designs generally assume instant reconfiguration, neglect the mechanical latency, or keep the antenna fixed. This paper proposes S-PASS, a star architecture of waveguides sharing one radio-frequency chain through a switch, and pipelined antenna pre-positioning (PAP), a time-division protocol in which idle waveguides pre-position their antennas for their next users while one transmits. It derives in closed form the scheduling efficiency, the number of waveguides beyond which the latency costs no transmission time, and the latency up to which PAP outperforms fixed placement. Analytical and numerical results show that this latency grows with each added waveguide, and users at the outer ends of the waveguides gain most.
\end{abstract}


\section{Introduction}

Pinching-antenna systems (PASS)~\cite{Ref_fukuda2022ntt,Ref_ding2025pass,Ref_liu2025tutorial} apply a dielectric particle, a ``pinch'', to a waveguide to create a radiation point, a \emph{pinching antenna} (PA), within meters of a user. Large-scale path loss and line-of-sight blockage become design variables. For multi-user service on a single waveguide, \emph{re-pinching}, i.e., releasing the pinch and re-attaching it at the next user's position, enables time-division multiple access (TDMA) with every user served from the point on the waveguide closest to it. In the symmetric-rate two-user example of~\cite{Ref_ren2025nomaoroma}, such re-pinched TDMA needs less transmit power than fixed-placement power-domain non-orthogonal multiple access (NOMA), because a fixed placement penalizes the user farthest from the pinch.
This advantage assumes that re-pinching takes no time. In practice, releasing the pinch and re-attaching it elsewhere is a mechanical cycle that can take milliseconds to seconds. A recent survey of reconfigurable antennas names the mismatch between that timescale and the much faster scheduling loop ``a key unresolved issue''~\cite{Ref_zhao2026rasurvey}. A re-pinch of tens of milliseconds can leave the channel idle for many slots. Re-pinched TDMA incurs this overhead between slots requiring different PA positions, whereas fixed-placement NOMA avoids it.

Existing designs treat reconfiguration latency in different ways. Waveguide switching and multiplexing studies~\cite{Ref_zhu2026ws}, the access protocols of segmented waveguide-enabled PASS (SWAN)~\cite{Ref_ouyang2025swan,Ref_shan2026swanprot}, and a switched-feed design that routes one radio-frequency (RF) chain among star-arranged waveguide segments to reduce terahertz in-guide loss~\cite{Ref_jiang2026thzpass} model at most the electronic switching time, not the mechanical latency. On a single waveguide, the latency of re-pinched TDMA has been modeled as a guard interval~\cite{Ref_jiang2026masurvey}, and antenna roaming accepts a finite moving speed and lets the antenna transmit while it travels~\cite{Ref_wang2026roaming}. Movable-antenna designs trade movement time against transmission time~\cite{Ref_hu2026matradeoff}, and electronically reconfigurable pinching avoids movement but radiates only where elements were installed~\cite{Ref_plata2026epinch}. None of them lets an idle PA move while another waveguide transmits, the case analyzed in this paper.

To fill this gap, this paper proposes S-PASS, a star architecture in which several dielectric waveguides, called \emph{arms}, extend radially from one central hub and share its single RF chain through a switch. On these arms, \emph{pipelined antenna pre-positioning} (PAP), a TDMA protocol, moves the antennas of the idle arms to their next users while one arm transmits, so that motion and transmission overlap as in a processor pipeline. The S in S-PASS stands for both features, the \emph{star} topology and the \emph{scheduling} it enables. Both features set S-PASS apart from center-fed PASS (C-PASS)~\cite{Ref_gan2026cpass}, which splits the power at the center of a single waveguide so that both halves radiate at once, and which does not consider antenna movement time. PAP absorbs the latency completely once enough arms share a balanced load and partly otherwise, and it reduces to conventional TDMA as the latency vanishes. The main contributions are as follows.
\begin{itemize}
\item \emph{We propose} the S-PASS architecture and the PAP protocol. S-PASS feeds all arms at the central hub, where one RF chain reaches them through a switch, whereas SWAN feeds its segments at remote points along the waveguide, which requires cabling to each feed. An isolated arm failure removes only its own sector. Each arm carries a single PA, so the pre-positioned antennas leave the uplink free of the inter-antenna radiation (IAR) that couples several active pinches on one waveguide~\cite{Ref_ouyang2025swan}.
\item \emph{We derive} the scheduling efficiency of PAP, the fraction of the frame that carries payload, over finite frames in closed form for its interleaved and burst schedules, the number of arms beyond which the latency costs no transmission time, and the latency up to which PAP outperforms the fixed-placement NOMA reference and fixed-placement TDMA. The analysis assumes quasi-static users, with balanced arm loads and equal movement times for the interleaved schedule.
\item \emph{We show} that this tolerable latency increases by the same amount with each additional arm. This holds for a fixed number of users per arm and a given fixed-placement baseline. Re-pinching improves fifth-percentile and mean spectral efficiency over fixed-placement TDMA once enough arms share the load in the evaluated free-space setting.
\end{itemize}

The rest of this paper is organized as follows. Section~\ref{sec:model} presents the system model, Section~\ref{sec:protocol} describes PAP, Sections~\ref{sec:theory} and~\ref{sec:cross} give the analysis and numerical results, and Section~\ref{sec:concl} concludes.

\section{System Model}\label{sec:model}

\subsection{Star Architecture}
The hub sits at the center of the coverage area, at $[0,0,h]^{T}$, i.e., at height $h$ above the ground plane. It hosts one RF chain that connects through an electronic $1{:}S$ switch to $S$ separate dielectric waveguides, the arms. Each arm is a complete waveguide with its own feed at the hub, not a segment of one shared waveguide as in SWAN~\cite{Ref_ouyang2025swan}. Two collinear arms resemble a C-PASS~\cite{Ref_gan2026cpass} whose center feed is switched between the two halves of the waveguide instead of splitting the power between them; more arms cover a two-dimensional area. The system operates in time-division duplex (TDD), so the one chain and one switch serve both link directions.

The arms point outward at distinct azimuths $\varphi_1<\dots<\varphi_S$ (Fig.~\ref{fig:sys}). Arm $s\in\{1,\dots,S\}$ extends from the hub to the coverage radius $D$ along the unit direction $\mathbf{e}_s=[\cos\varphi_s,\sin\varphi_s,0]^{T}$ and carries a single movable PA driven by its own linear actuator. All arms lie in the horizontal plane at height $h$, so the PA of arm $s$ is located at $\boldsymbol{\psi}_s=\rho_s\,\mathbf{e}_s+[0,0,h]^{T}$, $\rho_s\in[0,D]$, with fixed azimuth $\varphi_s$ and movable radius $\rho_s$. The single-axis antenna coordinate of the PASS literature~\cite{Ref_ding2025pass} thus becomes a polar one, and two-dimensional coverage decomposes into $S$ independent radial placement problems. The arms partition the disk of radius $D$ into \emph{nearest-arm sectors}: the sector of arm $s$ contains the ground points closest to arm $s$, with ties assigned to the lowest-index arm. There are $K$ single-antenna users on the ground plane, each assigned to the arm of its sector. Arm $s$ serves $K_s$ users, $\sum_s K_s=K$, the $k$th at the ground point $\mathbf{u}_{s,k}$. Its projection on the arm, the \emph{abscissa} $x_{s,k}=\min\{D,\max\{0,\mathbf{u}_{s,k}^{T}\mathbf{e}_s\}\}\in[0,D]$, is the closest feasible PA position to the user. At $\rho_s=x_{s,k}$, a horizontal offset $r_{s,k}=\|\mathbf{u}_{s,k}-x_{s,k}\mathbf{e}_s\|$ replaces $h$ by $\sqrt{h^2+r_{s,k}^2}$ below; the analysis takes $r_{s,k}=0$, which Fig.~\ref{fig:cov} relaxes.

\begin{figure}[!t]\centering
\includegraphics[width=0.80\columnwidth]{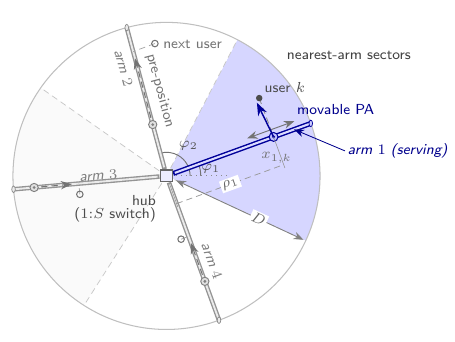}\\ {\footnotesize (a)}\\[5pt]
\includegraphics[width=0.85\columnwidth]{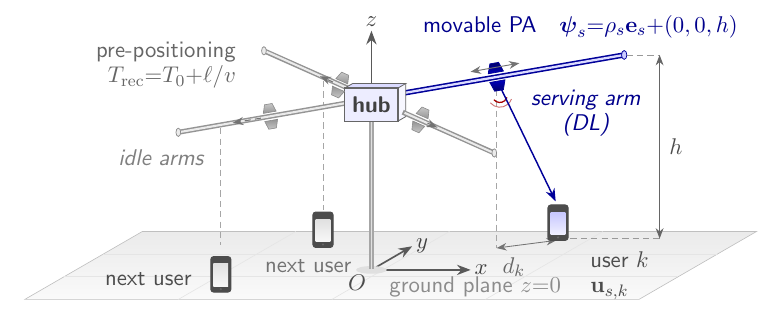}\\ {\footnotesize (b)}
\caption{S-PASS architecture. (a)~Bird's-eye view: nearest-arm sectors with one movable PA per arm. (b)~Perspective view: one RF chain behind a central $1{:}S$ switch; the active arm serves its user (solid), while idle arms pre-position their PAs (dashed).}\label{fig:sys}\end{figure}

\subsection{Channel and Signal Model}
Assuming line-of-sight serving links, we adopt the free-space model of~\cite{Ref_ding2025pass,Ref_liu2025tutorial}: the channel coefficient between a PA at $\boldsymbol{\psi}$ and a user at $\mathbf{u}$ is $h_{\mathrm o}(\mathbf{u},\boldsymbol{\psi})=\sqrt{\eta_0}\,e^{-j(2\pi/\lambda)\|\boldsymbol{\psi}-\mathbf{u}\|}/\|\boldsymbol{\psi}-\mathbf{u}\|$, $\eta_0=(\lambda/(4\pi))^2$, with $\lambda$ the carrier wavelength. Inside arm $s$ the guided mode travels from the feed to the PA over the radius $\rho_s$ and accumulates~\cite{Ref_ouyang2025swan}
\begin{equation}\label{eq:hin}
h_{\mathrm i}(\rho_s)=10^{-\frac{\alpha_g\rho_s}{20}}\,e^{-j\frac{2\pi\rho_s}{\lambda_g}},\qquad \lambda_g=\lambda/n_{\mathrm{eff}},
\end{equation}
where $n_{\mathrm{eff}}$ is the effective refractive index and $\alpha_g$ the attenuation in dB/m. We focus on the millimeter-wave (mmWave) regime and neglect in-guide loss in the rate analysis, setting $\alpha_g=0$ and $|h_{\mathrm i}|=1$; the terahertz regime, where this loss can be substantial~\cite{Ref_jiang2026thzpass}, is left to future work. Switch insertion loss is also neglected.

\subsubsection{Downlink}
In a downlink (DL) slot in which the switch connects the chain to arm $s$, its user $k$ receives $y_k=\sqrt{P_d}\,h_{\mathrm i}(\rho_s)\,h_{\mathrm o}(\mathbf{u}_{s,k},\boldsymbol{\psi}_s)\,b_k+n_k$, where $b_k$ is the unit-power data symbol for user $k$, $P_d$ the DL transmit power, and $n_k\sim\mathcal{CN}(0,\sigma_n^2)$ the noise. With the along-arm PA--user distance $d_k=|\rho_s-x_{s,k}|$, user $k$ achieves the rate
\begin{equation}\label{eq:rate}
R_k=\log_2(1+\gamma_k),\qquad \gamma_k=\frac{P_d\eta_0}{\sigma_n^{2}(d_k^{2}+h^{2})},
\end{equation}
where $\gamma_k$ is its signal-to-noise ratio. Both are maximized when the PA is directly over the user ($d_k=0$), where they no longer depend on $k$: $\gamma_k=\gamma_0\triangleq P_d\eta_0/(\sigma_n^{2}h^{2})$ and $R_k=R_0\triangleq\log_2(1+\gamma_0)$ for all $k$. Per-user re-pinching attains $R_0$ for every user.

\subsubsection{Uplink}
In an uplink (UL) slot in which the switch connects the chain to arm $s$, its PA couples the signal of user $k$ into the waveguide, $y_s=\sqrt{P_u}\,h_{\mathrm i}(\rho_s)\,h_{\mathrm o}(\mathbf{u}_{s,k},\boldsymbol{\psi}_s)\,b_{k}+n_s$, with $P_u$ the UL transmit power and $n_s\sim\mathcal{CN}(0,\sigma_u^2)$ the noise at the hub receiver, whose power $\sigma_u^2$ differs from the noise power $\sigma_n^2$ of the user terminal and may include the switch insertion loss. Because each arm carries exactly one PA, the captured signal passes no other pinch on its way to the feed, so the IAR that complicates multi-PA uplink modeling~\cite{Ref_ouyang2025swan} is absent by construction. Idle arms are assumed to have matched terminations at the switch and their far ends. Residual scattering, inter-arm coupling, and switch leakage are neglected. UL slots follow the same timing rules as DL slots, so every scheduling result below applies to either direction; the rate analysis considers the DL.

\section{Pipelined Antenna Pre-Positioning}\label{sec:protocol}

Time is partitioned into frames of duration $T_{\mathrm{frame}}$, in each of which every user is served once in a payload slot of duration $T_{\mathrm{slot}}$, the per-user service period of the scheduler, which may span many physical-layer slots. Switching the RF chain between arms costs an electronic switching time $T_{\mathrm{sw}}$, and re-pinching a PA over a distance $\ell$ costs a mechanical reconfiguration time $T_{\mathrm{rec}}(\ell)$, which can span many slots and is modeled as (A7) in \eqref{eq:Trec}. The analysis further assumes the following.
\begin{enumerate}[\IEEEsetlabelwidth{(A6)}]
\item[(A1)] \emph{Electronic switching}: $T_{\mathrm{sw}}\ll T_{\mathrm{slot}}$, and one guard of duration $T_{\mathrm{sw}}$ follows every payload slot, even on the same arm; the closing guard of a frame also covers the DL--UL transition.
\item[(A2)] \emph{Move--serve exclusion}: a PA cannot move during its own payload slot, whereas PAs of different arms move in parallel.
\item[(A3)] \emph{Stationarity}: user positions and arm assignments remain fixed throughout the analyzed frames; initialization is excluded from $T_{\mathrm{frame}}$.
\item[(A4)] \emph{Continuous-time slots}: a slot may start at any instant.
\item[(A5)] \emph{Lossless rate model}: $\alpha_g=0$ in \eqref{eq:hin}, and switch insertion loss is neglected.
\item[(A6)] \emph{Once-per-frame service}: each user is served exactly once per frame, i.e., round-robin TDMA, without which the converse in Theorem~\ref{thm:duty} fails.
\end{enumerate}

PAP has a single design parameter, the batch size $B$, and runs as follows: Step~1 once, Steps~2--4 in a loop within each frame, and Step~5 at each frame boundary.
\begin{enumerate}
\item \emph{Initialization.} Each arm $s$ sorts its $K_s$ users by abscissa into the sweep order $\pi_s$, a permutation of $\{1,\dots,K_s\}$ with $x_{s,\pi_s(1)}\le\dots\le x_{s,\pi_s(K_s)}$, and positions its PA at the first target, $\rho_s=x_{s,\pi_s(1)}$. Empty arms are skipped; the scheduling analysis assumes $K_s\ge1$ for all $S$ arms. By (A3), this step is amortized over many frames.
\item \emph{Data transmission.} The switch connects the RF chain to the active arm $s$ (A1), which serves up to $B$ of its remaining users in sweep order, one payload slot each, with $\rho_s=x_{s,k}$, the placement optimum of \eqref{eq:rate}. For $B>1$, the moves within a batch cannot overlap with service on other arms.
\item \emph{Antenna pre-positioning.} Meanwhile, every idle arm $j\ne s$ moves its PA to its next target, $\rho_j\to x_{j,k'}$, as soon as its own payload has ended (A2), so that the move overlaps the service of other arms and the electronic guards.
\item \emph{Arm selection.} The scheduler visits the arms in a fixed cyclic order, skipping arms whose users have all been served, until every user has been served exactly once (A6). Let payload $n=1,\dots,K$ of the frame be served on arm $s_n$ and start at $t_n$, with $t_1=0$. The PA of arm $s_n$ is ready at $r_n=0$ if this is the arm's first payload of the frame, since Step~1 or Step~5 has placed it at the target, and otherwise at $r_n=t_m+T_{\mathrm{slot}}+T_{\mathrm{rec}}(\ell_n)$, where $m<n$ indexes the arm's previous payload and $\ell_n$ is the distance between the two targets. Then
\begin{equation}\label{eq:timing}
t_n=\max\{t_{n-1}+T_{\mathrm{slot}}+T_{\mathrm{sw}},\,r_n\},
\end{equation}
and $T_{\mathrm{frame}}=t_K+T_{\mathrm{slot}}+T_{\mathrm{sw}}$ includes the closing guard. A move with $r_n\le t_{n-1}+T_{\mathrm{slot}}+T_{\mathrm{sw}}$ is \emph{absorbed}: the re-pinch still occurs but adds nothing to the frame length.
\item \emph{Sweep reversal.} At the frame boundary each arm reverses its sweep direction (the \emph{alternating-sweep} rule). Its PA thus ends the frame at one end of its users' span, exactly where the next frame begins, so no approach move is needed (Lemma~\ref{lem:sweep}), at a delay cost given in Remark~\ref{rem:delay}. Frame duration may depend on the sweep direction. Under TDD, DL and UL frames alternate at turnaround.
\end{enumerate}

In Fig.~\ref{fig:prtdma}, with $S{=}3$ arms and two users each, arm~1 re-pinches toward its second user $u_{1,2}$ while arms~2 and~3 are served. When the switch returns to arm~1, its PA is already in place. The batch size selects the schedule: \emph{interleaved} ($B{=}1$, one user per visit) or \emph{burst} ($B\ge\max_s K_s$, all users of an arm in one block), whose efficiencies the next section derives.

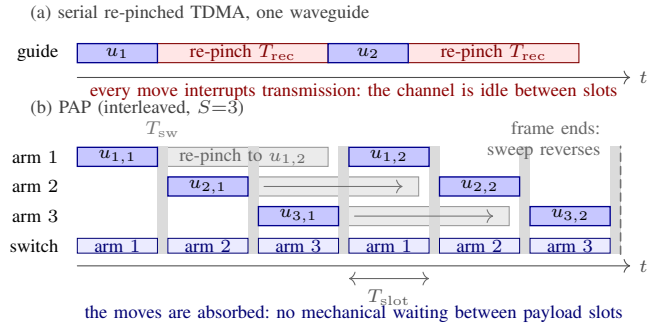
\begin{figure}[!t]\centering
\begin{tikzpicture}[x=1.06cm,y=0.28cm,font=\scriptsize,line cap=round]
\node[anchor=west,black!75] at (-0.70,10.6) {(a) serial re-pinched TDMA, one waveguide};
\node[anchor=east] at (-0.12,8.75) {guide};
\foreach \x in {1.0,4.13}{%
  \fill[red!10] (\x,8.3) rectangle ({\x+2.13},9.2);
  \draw[red!50!black,line width=0.4pt] (\x,8.3) rectangle ({\x+2.13},9.2);
  \node[red!50!black] at ({\x+1.065},8.75) {re-pinch $T_{\mathrm{rec}}$};}
\foreach \x/\lab in {0/$u_1$,3.13/$u_2$}{%
  \fill[blue!22] (\x,8.3) rectangle ({\x+1.0},9.2);
  \draw[blue!55!black,line width=0.5pt] (\x,8.3) rectangle ({\x+1.0},9.2);
  \node at ({\x+0.5},8.75) {\lab};}
\draw[->,black!70] (0,7.6) -- (6.9,7.6) node[right] {$t$};
\node[red!55!black,anchor=north,inner sep=1.5pt] at (3.45,7.55) {every move interrupts transmission: the channel is idle between slots};
\node[anchor=west,black!75] at (-0.70,6.15) {(b) PAP (interleaved, $S{=}3$)};
\fill[black!8] (1.00,3.4) rectangle (3.13,4.3);\draw[black!40,line width=0.4pt] (1.00,3.4) rectangle (3.13,4.3);
\fill[black!8] (2.13,2.0) rectangle (4.26,2.9);\draw[black!40,line width=0.4pt] (2.13,2.0) rectangle (4.26,2.9);
\fill[black!8] (3.26,0.6) rectangle (5.39,1.5);\draw[black!40,line width=0.4pt] (3.26,0.6) rectangle (5.39,1.5);
\foreach \i in {1,2,3,4,5,6}{\fill[black!14] ({\i*1.13-0.13},-0.75) rectangle ({\i*1.13},4.35);}
\node[black!55,anchor=south] at (1.065,4.35) {$T_{\mathrm{sw}}$};
\node[anchor=east] at (-0.12,3.85) {arm $1$};
\node[anchor=east] at (-0.12,2.45) {arm $2$};
\node[anchor=east] at (-0.12,1.05) {arm $3$};
\node[anchor=east] at (-0.12,-0.35) {switch};
\foreach \i/\yy/\lab in {0/3.4/$u_{1,1}$,3/3.4/$u_{1,2}$,1/2.0/$u_{2,1}$,4/2.0/$u_{2,2}$,2/0.6/$u_{3,1}$,5/0.6/$u_{3,2}$}{
  \fill[blue!22] ({\i*1.13},\yy) rectangle ({\i*1.13+1.0},{\yy+0.9});
  \draw[blue!55!black,line width=0.5pt] ({\i*1.13},\yy) rectangle ({\i*1.13+1.0},{\yy+0.9});
  \node at ({\i*1.13+0.5},{\yy+0.45}) {\lab};}
\node[black!60,inner sep=1pt] at (2.07,3.85) {re-pinch to $u_{1,2}$};
\draw[->,black!55] (2.33,2.45) -- (4.06,2.45);
\draw[->,black!55] (3.46,1.05) -- (5.19,1.05);
\foreach \i/\a in {0/1,1/2,2/3,3/1,4/2,5/3}{%
  \fill[blue!8] ({\i*1.13},-0.7) rectangle ({\i*1.13+1.0},0.0);
  \draw[blue!45!black,line width=0.35pt] ({\i*1.13},-0.7) rectangle ({\i*1.13+1.0},0.0);
  \node[blue!45!black] at ({\i*1.13+0.5},-0.35) {arm $\a$};}
\draw[densely dashed,black!55,line width=0.6pt] (6.78,-0.75) -- (6.78,4.35);
\node[black!60,anchor=south east,align=right,inner sep=2pt] at (6.60,3.48) {frame ends:\\sweep reverses};
\draw[->,black!70] (0,-1.3) -- (6.9,-1.3) node[right] {$t$};
\draw[<->,black!60] (3.39,-1.9) -- (4.39,-1.9);\node[black!65,anchor=north] at (3.89,-1.95) {$T_{\mathrm{slot}}$};
\node[blue!45!black,anchor=north,inner sep=1.5pt] at (3.45,-2.95) {the moves are absorbed: no mechanical waiting between payload slots};
\end{tikzpicture}
\caption{Timing of (a)~a serial re-pinched TDMA sequence on one waveguide and (b)~PAP with the interleaved schedule ($S{=}3$). Panel (b) shows one frame after initialization. Blue: payload slots; red: re-pinching that interrupts transmission; gray: re-pinching absorbed by the service of other arms and electronic guards; vertical bands: electronic guards $T_{\mathrm{sw}}$; dashed line: frame turnaround (Step~5).}
\label{fig:prtdma}
\end{figure}

\section{Performance Analysis}\label{sec:theory}

Since latency costs time, not signal quality, the figure of merit is the fraction of the frame that carries payload,
\begin{equation}
\eta\triangleq\frac{KT_{\rm slot}}{T_{\rm frame}}\in(0,1],
\end{equation}
termed the \emph{scheduling efficiency}, with $T_{\rm frame}$ from \eqref{eq:timing}; when every user is served at the common rate $R_0$, the effective sum rate within the frame is $R_{\rm eff}=\eta R_0$.

\subsection{Re-pinching Latency Model}
(A7) \emph{Latency model}: a PA move over a distance $\ell$ costs
\begin{equation}\label{eq:Trec}
T_{\mathrm{rec}}(\ell)=\begin{cases}
0, & \ell=0,\\
T_0+\ell/v, & \ell>0,
\end{cases}
\end{equation}
where the actuation constant $T_0$ lumps pinch release, acceleration, and settling, and $v$ is the travel speed. A recent survey places some PA reconfiguration operations at ``ms--s''~\cite{Ref_zhao2026rasurvey}, and the electronic-pinching proposal~\cite{Ref_plata2026epinch} does not report a switching time. We therefore treat \eqref{eq:Trec} as a parametric model with illustrative values $T_0\in[5,50]$\,ms and $v\in[0.05,0.5]$\,m/s, e.g., $7$\,ms for a 1\,mm adjustment and approximately $2$--$20$\,s for a 1\,m traverse. For arm $s$, whose $K_s$ users have $N_s$ distinct target positions spanning $D_s=\max_k x_{s,k}-\min_k x_{s,k}$, the total motion time per sorted sweep is $M_s\triangleq(N_s{-}1)T_0+D_s/v$.

\begin{lemma}[Sorted sweeps]\label{lem:sweep}
Under the alternating-sweep rule, with (A3) holding over two frames, arm $s$ travels exactly $D_s$ per frame, so its motion time per frame equals $M_s$, and no service order uses less motion time.
\end{lemma}

\begin{IEEEproof}
At least $N_s-1$ nonzero moves are needed to visit all $N_s$ distinct target positions. Any walk that visits all users of the arm travels at least $D_s$ per frame, with equality iff it is monotone between the two extremes. Starting at an interior position requires additional travel. The alternating-sweep rule avoids it, because each frame ends at the extreme where the next one starts.
\end{IEEEproof}

\subsection{Interleaved Schedule and the Sufficient Number of Arms}
With $B{=}1$ the switch cycles through arms $1,\dots,S$.

\begin{theorem}[Efficiency of the interleaved schedule]\label{thm:duty}
Consider $S$ arms with balanced loads $K_s=q=K/S$, where $q\ge1$ is an integer, and the sweep orders of Lemma~\ref{lem:sweep}. Let $T_{\rm rec}$ be the common duration of all consecutive-user transitions within a sweep. After initialization, the interleaved schedule attains the efficiency
\begin{equation}\label{eq:eta}
\begin{aligned}
\eta(S)=\min\!\bigg\{&\frac{T_{\rm slot}}{T_{\rm slot}+T_{\rm sw}},\\
&\frac{K\,T_{\rm slot}}{S(T_{\rm slot}+T_{\rm sw})+(q-1)(T_{\rm rec}+T_{\rm slot})}\bigg\},
\end{aligned}
\end{equation}
and no schedule respecting these sweep orders and satisfying (A1), (A2), and (A6) exceeds it. For $q>1$, the first term of \eqref{eq:eta}, the switching limit, is reached iff $S\ge S^\star$, where
\begin{equation}
S^\star=\Big\lceil\frac{T_{\rm rec}+T_{\rm slot}}{T_{\rm slot}+T_{\rm sw}}\Big\rceil.
\end{equation}
For $q=1$, no move is needed and the switching limit is always reached.
\end{theorem}

\begin{IEEEproof}
\emph{Achievability.} Visit the arms cyclically; by (A1) every service is followed by a guard of duration $T_{\rm sw}$. By \eqref{eq:timing}, without mechanical waiting, between two consecutive services of arm $s$, the arm is idle for $W(S)=(S-1)(T_{\rm slot}+T_{\rm sw})+T_{\rm sw}$, during which its PA may move (A2). Now $W(S)\ge T_{\rm rec}$ iff $S\ge(T_{\rm rec}+T_{\rm slot})/(T_{\rm slot}+T_{\rm sw})$, i.e., iff $S\ge S^\star$, and the frame is then switch-limited with $\eta=T_{\rm slot}/(T_{\rm slot}+T_{\rm sw})$. The first cycle of $S$ services needs no move because initialization has prepared its targets. If $S<S^\star$, consecutive cycles start $T_{\rm rec}+T_{\rm slot}$ apart, which (A4) permits. There are $q-1$ such gaps, and the last cycle lasts $S(T_{\rm slot}+T_{\rm sw})$, including its closing guard. Dividing $KT_{\rm slot}$ by the resulting frame duration gives \eqref{eq:eta}.
\emph{Converse.} The single RF chain carries $K$ services and $K$ guards per frame, so $T_{\rm frame}\ge K(T_{\rm slot}+T_{\rm sw})$. The last arm to begin service starts no earlier than $(S-1)(T_{\rm slot}+T_{\rm sw})$. Its $q$ services require $q-1$ start-to-start intervals of at least $T_{\rm slot}+T_{\rm rec}$ by (A2), followed by its final payload and guard. Hence, $T_{\rm frame}\ge S(T_{\rm slot}+T_{\rm sw})+(q-1)(T_{\rm slot}+T_{\rm rec})$. Dividing $KT_{\rm slot}$ by the larger bound gives \eqref{eq:eta}.
\end{IEEEproof}

The second term of \eqref{eq:eta} accounts for the $q-1$ moves per arm. For fixed $K$, increasing $S<S^\star$ raises efficiency, but generally not linearly.

\begin{remark}[Exactness of Theorem~\ref{thm:duty}]
Under the alternating-sweep rule, the last user served on an arm in one frame is served first on that arm in the next, so no move is needed between these two services. Theorem~\ref{thm:duty} counts this saving exactly. For $q>1$, the saving raises the efficiency by at most a factor $q/(q-1)$ relative to its large-$q$ limit at fixed $S$ and timing parameters, and it vanishes at the switching limit.
\end{remark}

\begin{remark}[Inter-service intervals]\label{rem:delay}
For the balanced interleaved schedule with equal move times, the maximum start-to-start interval between consecutive services of a user is $2T_{\rm frame}-S(T_{\rm slot}+T_{\rm sw})$. Each user still receives two slots per two frames. With alternating DL and UL frames of equal duration, the same-direction start-to-start interval is $2T_{\rm frame}$. The corresponding end-to-start intervals are smaller by $T_{\rm slot}$.
\end{remark}

\begin{corollary}[Limits]
For the same total $K$ and timing parameters, the efficiency of serial re-pinched TDMA is
\[
\eta(1)=\frac{K\,T_{\rm slot}}{KT_{\rm slot}+T_{\rm sw}+(K-1)\max\{T_{\rm sw},T_{\rm rec}\}}.
\] As $T_{\rm rec}\to0$, $\eta(S)\to T_{\rm slot}/(T_{\rm slot}+T_{\rm sw})$, the bound of ideal waveguide switching at zero reconfiguration cost~\cite{Ref_zhu2026ws}, which PAP attains even at nonzero latency when $q=1$ or $S\ge S^\star$.
\end{corollary}

At $T_{\rm rec}=50$\,ms and $T_{\rm sw}\approx0$, a 1\,ms slot requires an impractical $S^\star=51$ arms, whereas a 10\,ms slot needs $S^\star=6$. Under \eqref{eq:Trec} with $T_0=5$\,ms and $v=0.5$\,m/s, a 50\,ms re-pinch moves the PA by about 2\,cm, so these numbers apply to closely spaced consecutive target positions. A traverse of several meters takes seconds; absorbing it with few arms requires service periods of comparable length, as for quasi-static users.

\subsection{Burst Schedule}
With $B\ge\max_s K_s$ an arm serves all its users in one block before the switch moves on.

\begin{theorem}[Efficiency of the burst schedule]
For any $S\ge1$, the burst schedule achieves
\begin{equation}\label{eq:burst}
\eta_{\rm burst}=\frac{K\,T_{\rm slot}}{K(T_{\rm slot}+T_{\rm sw})+\sum_{s,j}\bigl[T_{\rm rec}(\ell_{s,j})-T_{\rm sw}\bigr]^+},
\end{equation}
where $[a]^+=\max\{a,0\}$ and $\ell_{s,j}=|x_{s,\pi_s(j+1)}-x_{s,\pi_s(j)}|$ is the distance between consecutive target positions. The sum is over all arms and $j=1,\ldots,K_s-1$.
\end{theorem}

\begin{IEEEproof}
Each arm sweeps once per frame, and its turnaround needs no approach move (Lemma~\ref{lem:sweep}). By \eqref{eq:timing}, each within-arm transition lasts $\max\{T_{\rm sw},T_{\rm rec}(\ell_{s,j})\}$, since motion overlaps the guard after the preceding payload. Each block also includes a final guard. Summing the $K$ payload durations and these intervals over all arms gives \eqref{eq:burst}.
\end{IEEEproof}

If every nonzero move lasts at least $T_{\rm sw}$, the denominator in \eqref{eq:burst} reduces to $KT_{\rm slot}+[K-\sum_s(N_s-1)]T_{\rm sw}+\sum_s M_s$. The individual move times then enter only through the total motion times $M_s$. For balanced loads with equal moves the interleaved schedule is at least as efficient (Theorem~\ref{thm:duty}).

\subsection{Rate Analysis Against Fixed Placement}
We compare PAP with fixed-placement two-user NOMA on one waveguide and TDMA on the same star, both of which avoid re-pinching. For users directly below their assigned arms, PAP serves every user at $R_0$ (A5) for a fraction $\eta(S)/K$ of each DL frame, an effective sum rate of $\eta(S)R_0$. Let $c\in(0,1]$ denote the ratio of a baseline's effective sum rate over DL time to $R_0$, including any electronic guards. The schemes use the same DL/UL time allocation, so PAP at least matches the baseline iff $\eta(S)\ge c$.

\begin{corollary}[Latency threshold]\label{cor:noma-cross}
Under the assumptions of Theorem~\ref{thm:duty}, let $q>1$. PAP at least matches a fixed-placement baseline of rate ratio $c$ iff $c\le T_{\rm slot}/(T_{\rm slot}+T_{\rm sw})$ and $T_{\rm rec}\le T_{\rm rec}^{*}(S)$, where
\begin{equation}\label{eq:trec-cross}
T_{\mathrm{rec}}^{*}(S)=\frac{S\bigl[qT_{\rm slot}/c-(T_{\rm slot}+T_{\rm sw})\bigr]}{q-1}-T_{\rm slot}.
\end{equation}
Strict outperformance holds iff both inequalities are strict. For $q=1$, no movement is needed and comparison depends only on the switching limit.
\end{corollary}

The result follows by requiring both terms in \eqref{eq:eta} to be at least $c$. The tolerable latency increases by the same amount per added arm when $q$, $c$, $T_{\rm slot}$, and $T_{\rm sw}$ are fixed. The rate ratios of the two baselines are evaluated for $h=3$\,m, $D=10$\,m, and $\gamma_0=39.1$\,dB (28\,GHz, $P_d=20$\,dBm, $\sigma_n^2=-90$\,dBm), so that $R_0=12.98$\,bit/s/Hz.

\subsubsection{Fixed-Placement NOMA}
Consider two users directly below one arm at spacing $\Delta=5$\,m with symmetric rate targets, motivated by~\cite{Ref_ren2025nomaoroma}. The NOMA reference keeps the pinch fixed and uses superposition coding~\cite{Ref_ding2025pass} to serve both users continuously during DL time, neglecting guard overhead. Its sum rate is used as a reference for PAP with $K=Sq$ users. A scan over placements shows that placing the pinch over one user is nearly optimal for these parameters. This near user removes the other user's signal by perfect successive interference cancellation and sees $\gamma_0$, while the far user sees $\gamma_\Delta=P_d\eta_0/[\sigma_n^2(h^2+\Delta^2)]$. With a fraction $\alpha$ of the power for the near user, equal rates require $(1+\alpha\gamma_0)(1+\alpha\gamma_\Delta)=1+\gamma_\Delta$, whose positive root gives $c_{\rm N}=2\log_2(1+\alpha^\star\gamma_0)/\log_2(1+\gamma_0)$, with
\[
\alpha^\star=\frac{\sqrt{(\gamma_0+\gamma_\Delta)^2+4\gamma_0\gamma_\Delta^2}-(\gamma_0+\gamma_\Delta)}{2\gamma_0\gamma_\Delta}.
\]
Here $c_{\rm N}=0.9966$: NOMA at this placement achieves a sum rate $0.34\%$ below $R_0$, so PAP exceeds this reference rate only near the switching limit. Matching this reference rate requires $c_{\rm N}\le T_{\rm slot}/(T_{\rm slot}+T_{\rm sw})$, i.e., $T_{\rm sw}\lesssim3.4$\,\textmu s per 1\,ms slot.

\subsubsection{Fixed-Placement Star TDMA}
This baseline fixes the pinch of each arm at the arm center, the fixed placement maximizing the expected sum rate for users uniform on $[0,D]$, and incurs only $T_{\rm sw}$ per slot. Averaging \eqref{eq:rate} over a user uniform on the arm, at distance up to $m=D/2$ from the pinch, gives the sum-rate ratio before guard overhead
\begin{equation}\label{eq:rho-place}
\rho_{\mathrm{place}}=\frac{\ln\frac{m^2+a^2}{m^2+h^2}+\frac{2a}{m}\arctan\frac{m}{a}-\frac{2h}{m}\arctan\frac{m}{h}}{\ln(1+\gamma_0)},
\end{equation}
with $a^2=h^2+P_d\eta_0/\sigma_n^{2}$. The effective sum-rate ratio is $c=\rho_{\rm place}T_{\rm slot}/(T_{\rm slot}+T_{\rm sw})$. The user at the outer arm endpoint, at distance $m$ from the fixed pinch, has the payload-rate ratio $\epsilon_{\rm edge}=\log_2\!\big(1+P_d\eta_0/[\sigma_n^2(h^2+m^2)]\big)/R_0$. For this user's throughput, Corollary~\ref{cor:noma-cross} applies with $c=\epsilon_{\rm edge}T_{\rm slot}/(T_{\rm slot}+T_{\rm sw})$, since the common factor $R_0/K$ cancels. Here $\rho_{\rm place}=0.937$ and $\epsilon_{\rm edge}=0.852$ for every $S$, because every arm spans $D$. Once the latency is absorbed, the common guard factor cancels, so re-pinching gains $6.7\%$ in expected sum rate and $17\%$ for the outer-endpoint user.

Table~\ref{tab:cross} lists the resulting thresholds for $q=2$, equal move times, and $T_{\rm sw}/T_{\rm slot}=10^{-4}$. For users 5\,m apart, a move takes about 10\,s with $T_0=5$\,ms and $v=0.5$\,m/s. Exceeding the NOMA reference even with eight arms then requires service periods longer than about $1.4$\,s, as for quasi-static users.
\begin{table}[!t]
\centering
\caption{Latency threshold $T^{*}_{\rm rec}(S)/T_{\rm slot}$ from \eqref{eq:trec-cross} up to which PAP at least matches a fixed-placement reference of rate ratio $c$}
\label{tab:cross}
\setlength{\tabcolsep}{4.5pt}
\begin{tabular}{lccccc}
\hline
Reference & $c$ & $S{=}1$ & $S{=}2$ & $S{=}4$ & $S{=}8$\\
\hline
NOMA, sum rate & 0.9966 & 0.0067 & 1.01 & 3.03 & 7.05\\
Star TDMA, sum rate & 0.9371 & 0.134 & 1.27 & 3.54 & 8.07\\
Star TDMA, outer endpoint & 0.8522 & 0.347 & 1.69 & 4.39 & 9.77\\
\hline
\end{tabular}
\end{table}
For balanced loads and equal move times, use the interleaved schedule with $S\ge S^\star$ to absorb the latency (Theorem~\ref{thm:duty}). Its gain over fixed-placement TDMA is largest for the outer-endpoint user.

\section{Numerical Results}\label{sec:cross}

A simulator first validates the closed-form laws and then relaxes the assumption that users lie directly below the arms. It executes the timing recursion \eqref{eq:timing} in continuous time (A4). The parameters are those of the rate analysis, with $T_{\rm slot}=10$~ms, $T_{\rm sw}=1$~\textmu s, and $K_s=4$ equispaced users per arm. To sweep the latency we fix $T_0=5$~ms and $v=0.5$~m/s and vary the positive user spacing $\ell$ up to 24.75~cm, so that $T_{\rm rec}(\ell)$ in \eqref{eq:Trec} ranges from $0.5001$ to $50$ slot durations. The zero-distance case, $T_{\rm rec}(0)=0$, is checked separately.
\begin{figure}[!t]
\centering
\includegraphics[width=0.73\columnwidth]{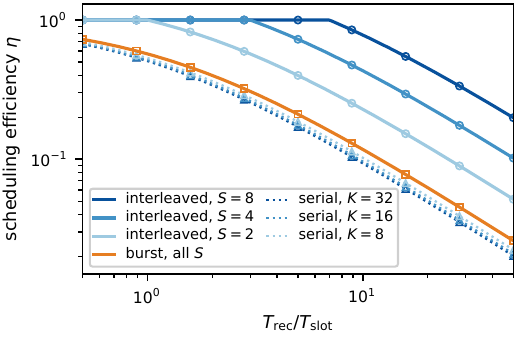}
\caption{Scheduling efficiency versus $T_{\rm rec}/T_{\rm slot}$, at $K_s{=}4$ users per arm and $T_{\rm slot}{=}10$~ms. Serial curves use the same total $K=4S$. Lines: from \eqref{eq:eta}, \eqref{eq:burst}, and $\eta(1)$; markers: simulation.}
\label{fig:eff}
\end{figure}
\begin{figure}[!t]
\centering
\includegraphics[width=0.73\columnwidth]{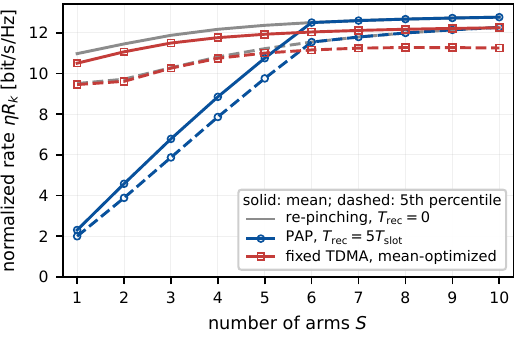}
\caption{Mean (solid) and 5th-percentile (dashed) normalized rate $\eta R_k$ from $10^5$ uniform disk samples. PAP: conditional on $q=4$ and $T_{\rm rec}=5T_{\rm slot}$ in \eqref{eq:eta}; re-pinching with $T_{\rm rec}=0$: zero-latency reference; fixed-placement TDMA: mean-rate-maximizing pinch position.}
\label{fig:cov}
\end{figure}

In Fig.~\ref{fig:eff}, interleaved and burst markers agree with \eqref{eq:eta} and \eqref{eq:burst}, respectively, to within $10^{-12}$. Interleaving reaches the switching limit at $S\ge S^\star$. With equal moves and fixed $K_s$, burst does not depend on $S$ and stays below interleaving over the plotted range. For the rate ratios, a Monte Carlo average over 2000 drops of 24 uniform users per arm reproduces \eqref{eq:rho-place} to within $2\times10^{-4}$, and the worst user of each drop averages a ratio of 0.859, approaching $\epsilon_{\rm edge}=0.852$. Fig.~\ref{fig:cov} evaluates the mean and 5th percentile of $\eta R_k$, i.e., $K$ times per-user throughput within a DL frame. Spatial averaging uses $10^5$ uniform disk samples, and the efficiency of PAP follows from \eqref{eq:eta} with $q=4$ and $T_{\rm rec}=5T_{\rm slot}$. Fixed-placement TDMA maximizes the sample mean rate over a 0.05~m grid. PAP absorbs the latency at $S^\star=6$. Its 5th-percentile gain over fixed placement increases from 3.5\% at six arms to 9\% at ten, while the mean gain stays near 4\%. More arms narrow the sectors and shrink the lateral offsets that both schemes share, whereas only PAP removes the along-arm distance.

A separate test uses 1000 independent drops of $K=24$ users uniform in the disk, with the same $T_0$ and $v$, distance-based move times, and empty arms skipped. With eight arms, the scheduling efficiency of the interleaved schedule over a DL--UL cycle averages 0.0116 at $T_{\rm slot}=10$~ms and 0.923 at $T_{\rm slot}=2$~s; the 95\% confidence intervals across drops are $[0.0114,0.0117]$ and $[0.920,0.926]$, respectively. These timelines include unequal arm loads and move times. Thus, meter-scale motion needs service periods of seconds; the 50~ms moves assumed in Fig.~\ref{fig:cov} apply to closely spaced users.

\section{Conclusion}\label{sec:concl}

S-PASS combines a star of waveguide arms on one RF chain with PAP, which pre-positions the antennas of idle arms while one arm transmits. Under balanced loads and equal move times, the scheduling efficiency increases with the number of arms up to the switching limit; six arms suffice for a 50~ms re-pinch and 10~ms slots. With $q=2$ and $T_{\rm sw}/T_{\rm slot}=10^{-4}$, PAP at least matches the two-user NOMA reference for latencies up to 0.0067 slot durations with one arm and 7.05 with eight, about one slot duration more per added arm. Once the latency is absorbed and users lie directly below the arms, the gains over fixed-placement TDMA are 6.7\% in expected sum rate and 17\% for the outer-endpoint user. With users spread over the disk and a move time of five slot durations, the 5th-percentile gain reaches 9\% at ten arms. Random user drops with distance-based move times show that meter-scale motion needs service periods of seconds. Future work includes measured actuator parameters and the terahertz regime.
\bibliographystyle{IEEEtran}
\bibliography{main}
\end{document}